\documentclass{sigplanconf}

\usepackage{lipsum}
\usepackage{amsfonts}
\usepackage{amsmath}
\usepackage{amsthm}
\usepackage{graphicx}
\usepackage{subfig}
\usepackage{epstopdf}
\usepackage{fancyvrb}
\usepackage{multicol}

\newtheorem{definition}{Definition}
\newtheorem{proposition}{Proposition}

\begin{document}

\conferenceinfo{PLDI'13,} {June 16--19, 2013, Seattle, WA, USA.}
\CopyrightYear{2013}
\copyrightdata{978-1-4503-2014-6/13/06}
\clubpenalty=10000 
\widowpenalty = 10000

\preprintfooter{short description of paper}   

\title{AutoSynch: An Automatic-Signal Monitor \\ Based on Predicate Tagging}

\authorinfo{Wei-Lun Hung \and Vijay K. Garg}
           {Department of Electrical and Computer Engineering \\
           The University of Texas at Austin}
           {wlhung@utexas.edu, garg@ece.utexase.edu}

\maketitle
\begin{abstract}
Most programming languages use monitors with {\em explicit signals}
for synchronization in shared-memory programs. Requiring programmers to signal
threads explicitly results in many concurrency bugs due to missed notifications, or
notifications on wrong condition variables.
In this paper, we describe an implementation of an automatic signaling monitor in Java
called {\em AutoSynch} that eliminates such concurrency bugs by 
removing the burden of signaling from the programmer. We show that the
belief that automatic signaling monitors are prohibitively expensive is wrong.
For most problems, programs based on {\em AutoSynch} are almost as fast as 
those based on explicit signaling. For some, {\em AutoSynch} is even faster than
explicit signaling because it never uses {\em signalAll}, whereas the programmers
end up using {\em signalAll} with the explicit signal mechanism.

{\em AutoSynch} achieves efficiency in synchronization based on three novel ideas.
We introduce an operation called
{\em globalization} that enables the predicate evaluation in
every thread, thereby reducing context switches during the execution of the program. 
Secondly, {\em AutoSynch} avoids {\em signalAll} by using a property called {\em relay invariance}
that guarantees that whenever possible there is always at least one thread whose condition is true which has been 
signaled.
Finally, {\em AutoSynch} uses a technique called {\em predicate tagging} to efficiently
determine a thread that should be signaled.
To evaluate the efficiency of {\em AutoSynch}, we have implemented many different 
well-known synchronization problems such as the producers/consumers problem,
the readers/writers problems, and the dining philosophers problem. The results
show that {\em AutoSynch} is almost as efficient as the explicit-signal monitor
and even more efficient for some cases. 

\end{abstract}

\category{D.1.3}{Concurrent Programming}{Parallel programming}
\category{D.3.3}{Language Constructs and Features}{Concurrent programming
structures; classes and objects; control structures}
\terms
Algorithms, Languages, Performance 

\keywords
automatic signal, explicit signal, implicit signal, monitor, 
concurrency, parallel


\section{Introduction} \label{sec:intro}

Multicore hardware is now ubiquitous. Programming these multicore processors is
still a challenging task due to
bugs resulting from concurrency and synchronization.
Although there is widespread acknowledgement of difficulties 
in programming these systems, it is surprising that by and large the most 
prevalent methods of dealing with synchronization are based on ideas that were 
developed in early 70's \cite{dijk68, hoa74, bh75a}. For 
example, the most widely used threads package in C++ \cite{stro97}, 
pthreads \cite{bute97}, and the most widely used threads package in Java \cite{gjs00}, 
java.util.concurrent \cite{lea05}, are based
on the notion of monitors \cite{hoa74, bh75a}(or semaphores 
\cite{dijk65, dijk68}). 
In this paper, we propose a new method called {\em AutoSynch} based on
automatic signaling monitor 
that allows gains in productivity of the programmer as well as gain in
performance of the system.

Both pthreads and Java require programmers to explicitly
signal threads that may be waiting on certain condition. The programmer
has to explicitly declare condition variables and then signal one
or all of the threads when the associated condition becomes true.
Using the wrong waiting notification ({\em signal} versus {\em signalAll} or notify
versus notifyAll) is a frequent source of bugs in Java multithreaded
programs. In our proposed approach, {\em AutoSynch}, there is no notion of condition variables
and it is the responsibility of the system to signal appropriate threads.
This feature significantly reduces the program size and complexity.
In addition, it allows us to completely eliminate signaling more than 
one thread resulting in reduced context switches and better performance.
The idea of automatic signaling was initially explored by Hoare in \cite{hoa74},
but rejected in favor of condition variables due to efficiency considerations.
The belief that automatic signaling is extremely inefficient compared to
explicit signaling is widely held since then and
all prevalent concurrent languages based on monitors use
explicit signaling.
For example, Buhr, Fortier, and Coffin claim that automatic monitors are $10$ to $50$ times
slower than explicit signals \cite{bfc95}. The reason for this drastic slowdown in 
previous implementations of automatic monitor  is that they evaluate
all possible conditions on which threads are waiting whenever the monitor
becomes available. We show in this
paper that the widely held belief is wrong.

 With careful analysis of the conditions on which
the threads are waiting and evaluating as few conditions as possible, automatic signaling can be as efficient 
as explicit signaling.
In {\em AutoSynch}, the programmer simply specifies the predicate $P$ on which the thread is 
waiting using the construct {\em waituntil(P)} statement. 
When a thread executes the statement, it checks whether $P$ is true. If it is true, the thread can continue; 
otherwise, the thread must wait for the system to signal it. The {\em
AutoSynch} system has a condition manager that is responsible for determining
which thread to signal by analyzing the predicates and the state of the shared object. 


Fig.~\ref{fig:bb_exp} shows the difference between 
the Java and the {\em AutoSynch} implementation for the parameterized bounded-buffer
problem, a variant bounded-buffer problem (also known as the producer-consumer problem) \cite{dijk65, dijk71}. 
In this problem, producers put items into the shared buffer, while 
consumers take items out of the buffer. The {\em put} function has a parameter
$items$; the {\em take} function has a parameter, {\em num}, 
indicating the number of items taken. There are two requirements for synchronization.
First, every operation on a shared variable, such as {\em buff}, should be done 
under mutual exclusion. Second, we need {\em conditional synchronization};
a producer must wait when the buffer has no sufficient space, and a consumer 
must wait when the buffer has no sufficient items. The explicit-signal 
bounded-buffer is written in Java. A lock variable and two associated condition 
variables are used to maintain mutual exclusion and conditional 
synchronization. A thread needs to acquire the lock before entering member
functions. In addition, programmers need to explicitly associate conditional 
predicates with condition variables and call {\em signal} ({\em signalAll}) or
{\em await} statement manually. Note that, the {\em unlock} statement should be 
done in a {\em finally} block, {\em try} and {\em catch} blocks are also need for the 
{\em InterruptedException} that may be thrown by {\em await}. However, for
simplicity, we avoid the exception handling in Fig.~\ref{fig:bb_exp}. 
The automatic-signal bounded-buffer is written using {\em AutoSynch} framework.
As in line $1$, we use {\em AutoSynch} modifier to indicate that the class is a 
monitor, all member functions of the class is mutual exclusion. For conditional 
synchronization,  we use {\em waituntil} as in line 9. There are 
no {\em signal} or {\em signalAll} calls in the {\em AutoSynch} program.
Clearly, the automatic-signal monitor is
much simpler than the explicit-signal monitor. 

\begin{figure*}[ht!]
\begin{multicols}{2}
    \begin{Verbatim}[fontsize=\footnotesize,gobble=8,frame=topline,
            framesep=5mm,numbers=left,numbersep=2pt,
            label=\fbox{\small\emph{Explicit-Signal}}]
        class BoundedBuffer {
          Object[] buff;  
          int putPtr, takePtr, count;
          Lock mutex = new ReentrantLock();
          Condition insufficientSpace = mutex.newCondition();
          Condition insufficientItem = mutex.newCondition();
          public BoundedBuffer(int n) {
            buff = new Object[n];
            putPtr = takePtr = count = 0;
          }
          public void put(Object[] items) {
            mutex.lock();
            while (items.length + count > buff.length) {
              insufficientSpace.await();
            }
            for (int i = 0; i < items.length; i++) {
              buff[putPtr++] = items[i];
              putPtr %= buff.length;
            }
            count += items.length;
            insufficientItem.signalAll();
            mutex.unlock();
          }
          public Object[] take(int num) {
            mutex.lock();
            while (count < num) {
              insufficientItem.await();
            }
            Object[] ret = new Object[num];
            for (int i = 0; i < num; i++) {
              ret[i] = buff[takePtr++];
              takePtr %= buff.length;
            }
            count -= num;
            insufficientSpace.signalAll();
            mutex.unlock();
            return ret;
          }
        }
    \end{Verbatim} 
    \begin{Verbatim}[fontsize=\footnotesize,gobble=8,frame=lines,framesep=5mm,
            numbers=left,numbersep=2pt,
            label=\fbox{\small\emph{Automatic-Signal}}]
        AutoSynch class BoundedBuffer { 
          Object[] buff; 
          int putPtr, takePtr, count; 
          public BoundedBuffer(int n) {
            buff = new Object[n];
            putPtr = takePtr = count = 0;
          }
          public void put(Object[] items) { 
            waituntil(count + items.length <= buff.length); 
            for (int i = 0; i < items.length; i++) {
              buff[putPtr++] = items[i];
              putPtr %= buff.length;
            }
            count += items.length; 
          } 
          public Object[] take(int num) { 
            waituntil(count >= num);
            Object[] ret = new Object[num];
            for (int i = 0; i < num; i++) {
              ret[i] = buff[takePtr++]; 
              takePtr %= buff.length; 
            }
            count -= num;
            return ret;
          }
        }
    \end{Verbatim}
\end{multicols}
    \caption{The parameterized bounded-buffer example}
    \label{fig:bb_exp}
\end{figure*}


\begin{figure}[ht!]
  \centering
  \includegraphics[width=70mm]{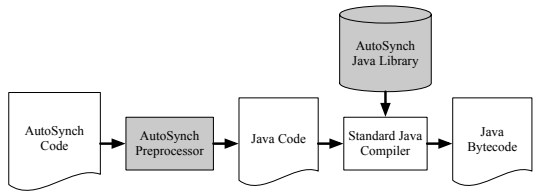}
  \caption{The framework of {\em  AutoSynch}}
  \label{fig:fw}
\end{figure}

To facilitate the automatic-signal mechanism in Java, we have implemented the 
framework of {\em AutoSynch} illustrated in Fig.~\ref{fig:fw}. 
The framework is composed of a 
preprocessor and a Java condition manager library. The preprocessor translates 
{\em AutoSynch} code into traditional Java code. 
Our automatic-signal mechanism and 
developed techniques were implemented in the Java condition manager library, 
which is responsible for monitoring the state of the monitor object
and signaling an appropriate thread.

%

In this paper, we argue that automatic signaling is generally as fast as explicit 
signaling (and even faster for some examples). In Section \ref{sec:sigAll}, we give reasons
for the efficiency of automatic signaling. In short, the explicit signaling has 
to resort to {\em signalAll} in some examples; however, our automatic signaling never 
uses {\em signalAll}. Thus {\em AutoSynch} is 
considerably faster for synchronization problems with 
{\em signalAll}. The design principle underlying {\em AutoSynch} is to reduce 
the number of context switches and predicate evaluations.

\begin{description}
    \item{\bf Context switch:} A context switch requires a certain 
        amount of time to save and load registers and update various tables and
        lists. Reducing unnecessary context switches boosts the performance of the system.
        A {\em signalAll} call introduces unnecessary context switches; therefore,
        {\em signalAll} calls are never used in {\em AutoSynch}. 
    \item {\bf Predicate evaluation:} In the 
        automatic-signal mechanism, signaling a thread is the 
        responsibility of the system. The number of predicate evaluations is 
        crucial for efficiency in deciding which thread should be signaled. 
        By analyzing the structure of the predicate, our system reduces the number of predicate evaluations.
\end{description}

There are three important novel concepts in {\em AutoSynch} that enables efficient automatic
signaling --- {\em globalization of predicates, relay invariance}, and {\em predicate
tagging}.

The technique of {\em globalization} of a predicate $P$ is used to reduce the number of 
context switches for its evaluation. In the current systems, only the thread that is waiting
for the predicate $P$ can evaluate it. When the thread is signaled, it wakes up, acquires the
lock to the monitor and then evaluates the predicate $P$. If the predicate $P$ is false, it
goes back to wait resulting in an additional context switch. In {\em AutoSynch} system, the thread
that is in the monitor evaluates the condition for the waiting thread and wakes it only
if the condition is true. Since the predicate $P$ may use variables local to the thread
waiting on it, {\em AutoSynch} system derives a globalization predicate $P'$ of
the predicate $P$, such that other threads can evaluate $P'$. 
The details of globalization are in Section \ref{sec:globalization}.


The idea of {\em relay invariance} is used to avoid {\em signalAll} calls in {\em
AutoSynch}.
The relay invariance ensures that if there is any thread whose waiting condition is true, then
there exists at least one thread whose waiting condition is  true and is signaled by the system.
With this invariance, the {\em signalAll} call 
is unnecessary in our automatic-signal mechanism. This mechanism reduces 
the number of context switches by avoiding {\em signalAll} calls. 
The details of this approach are in Section \ref{sec:relay}.

The idea of {\em predicate tagging} is used to accelerate the process of deciding which thread to signal.
All the waiting conditions are analyzed and tags are assigned to every predicate
according to its semantics. To decide which thread should be 
signaled, we identify tags that are most likely to be true after examining the 
current state of the monitor. Then we only evaluate the predicates with 
those tags. 
The details of predicate tagging are in Section \ref{sec:tag}.

Our experimental results indicate that {\em AutoSynch} can significantly improve
performance compared to other automatic-signal mechanisms \cite{bh05}. In \cite{bfc95,bh05}
the automatic-signal mechanism is 10-50 times
slower than the explicit-signal mechanism; however, {\em AutoSynch} is 
only 2.6 times slower than the explicit-signal mechanism even in the worst 
case of our experiment results. Furthermore, {\em AutoSynch} is 26.9
times faster than the explicit-signal mechanism in the parameterized
bounded-buffer problem that relies on 
{\em signalAll} calls. Besides, the experimental results also show that {\em AutoSynch} 
is scalable; the performance of {\em AutoSynch} is stable even if the number of 
threads increases for many problems conducted in the paper. 

Although the experiment results show that {\em AutoSynch} is 2.6 times slower than 
the explicit-signal mechanism in the worst case, it is still desirable to have automatic signaling.
First, automatic signaling simplifies the task of concurrent programming.
            In explicit-signal monitor, it is the responsibility of programmers to 
        explicitly invoke a {\em signal} call on some condition variable for
        conditional synchronization. Using the wrong notification, and signaling
        a wrong condition variable are frequent sources of bugs.  The idea is analogous to automatic
        garbage collection.
Although garbage collection leads to decreased
performance because of the overhead in deciding which memory to free, programmers 
avoid manual memory deallocation. As a consequence, memory leaks and certain 
bugs, such as dangling pointers and double free bugs, are reduced. 
Similarly, automatic-signal mechanism consumes computing resources 
in deciding which thread to be signaled; programmers avoid explicitly invoking 
{\em signal} calls. As a result, some bugs, such as using wrong notification and
signaling a wrong condition variable, are eliminated.  Secondly,
 in explicit-signal monitor, the principle of separation of concerns is violated.
    Any method that changes the state of the monitor must be aware of all the 
    conditions, which other threads could be waiting for, 
    in other methods of the monitor. 
     The intricate relation between
        threads for conditional synchronization breaks the modularity and 
        encapsulation of programming.  
   Finally,
  {\em AutoSynch} can provide  
        rapid prototyping  in developing programs and accelerating 
        product time to market. Moreover, 
        a correct
        automatic-signal implementation is helpful in debugging an
        explicit-signal implementation.

Although this paper 
focuses on Java, our techniques are also applicable to other programming 
languages and models, such as pthread and C\# \cite{hwg03}.

This paper is organized as follows. Section \ref{sec:bg} gives the background
of the monitor. 
Section \ref{sec:sigAll} explains why {\em signalAll} is required for
explicit-signal monitor but not automatic-signal monitor. The concepts of 
{\em AutoSynch} are presented in Section \ref{sec:concept} and the practical 
implementation details are discussed in Section  \ref{sec:imp}. The proposed 
methods are then evaluated with experiments in Section \ref{sec:eval}.  
Section \ref{sec:conclu} gives the concluding remarks.

\section{Background: monitor} \label{sec:bg} 
Monitor is an abstract object or module containing shared data to be used safely
by multiple threads in concurrent programming. Monitor can
be defined by two characteristics, mutual exclusion and conditional 
synchronization. Mutual exclusion guarantees that at most one thread can 
execute any member function of a monitor at any time.  Threads  acquire the 
lock of the monitor to acquire the privilege for accessing it. Conditional synchronization 
maintains the execution order between threads. Threads may wait for some 
condition to be met and release the monitor lock temporarily. After the 
condition has been met, threads then re-acquire the lock and continue to 
execute. According to Buhr and Harji \cite{bh05}, monitors can be divided into 
two categories according to the different implementations of conditional 
synchronization. 
\begin{description}
    \item{\bf Explicit-signal monitor} In this type of monitor, condition
        variables, {\em signal} and {\em await} statements are used for synchronization. 
    Programmers need to associate assertions with condition variables manually.
    This mechanism involves two or more threads. A thread  waits on some condition variable 
    if its predicate is not true. When another thread detects that the state has 
    changed and the predicate is true, it explicitly signals the 
    appropriate condition variable.
    \item{\bf Automatic-signal (implicit-signal) monitor} This kind of monitor 
    uses {\em waituntil}
    statements, such as line 9 in automatic-signal program in
    Fig.~\ref{fig:bb_exp}, instead of condition variables for
    synchronization. Programmers do not need to associate assertions with
    variables, but use {\em waituntil} statements directly. In
    monitor, a thread will wait as long as the condition of a {\em waituntil}
    statement is false, and execute the remaining tasks only after the condition 
    becomes true. The responsibility of signaling a waiting thread is that of 
    the system rather than of the programmers. 
\end{description}


\section{{\em signalAll} requirement in explicit} 
\label{sec:sigAll}
The {\em signalAll} call is essential in explicit-signal mechanism when programmers
do not know which thread should be signaled. In Fig.~\ref{fig:bb_exp}, a 
producer must wait if there is no space to put $num$ items, while a consumer 
has to wait when the buffer has insufficient items.
Since producers and consumers can put and take different numbers of items
every time, they may wait on different conditions to be met. Programmers 
do not know which producer or consumer should be signaled at runtime. 
Therefore, the {\em signalAll} call is used instead of {\em signal} calls in
line $21$ and $35$.
Although programmers can avoid using {\em signalAll} calls by writing
complicated code that associates different conditions to different condition 
variables; the complicated code makes the maintenance of the program
bad. 


The {\em signalAll} call is expensive; it may decreases the performance because 
it introduces redundant context switches, requiring 
computing time to save and load registers and update various tables and lists.
Furthermore, {\em signalAll} calls cannot increase parallelism because threads
are forbidden to access a monitor simultaneously. Although multiple threads are
signaled at a time, only one thread is able to acquire the monitor. Other 
threads may need to go back to waiting state since another thread may change 
the status of the monitor. Suppose in Fig.~\ref{fig:bb_exp}, the 
buffer has 64 items after a producer finishes a put call. The producer calls 
{\em insufficientItem.signalAll()} in line $21$ before completing the call.
$10$  
waiting consumers are signaled; each of them is waiting to take $48$ items.
Suppose the consumer $C$ re-acquires the lock first and takes $48$ items. The
remaining items, $16$, are insufficient for the other threads; they
make context switches, re-evaluate their predicates, and go back to waiting 
state. Theses context switches are redundant since the $9$ threads do not 
make any progress but only go back to waiting state. 
Therefore, if we avoid using the {\em signalAll} call and only signal a thread that 
is most likely to make progress, the unnecessary context switches can be
reduced.

\section{{\em AutoSynch} concepts} \label{sec:concept}

\subsection{Predicate evaluation} \label{sec:globalization}
In {\em AutoSynch}, it is the responsibility of the system to signal appropriate 
threads automatically. The predicate evaluation is crucial in deciding which
thread should be signaled. We discuss how to preform predicate evaluations
of {\em waituntil} statements. 

A predicate $P(\vec{x}): X \rightarrow \mathbb{B}$ is a Boolean condition, 
where $X$ is the space spanned by the variables $\vec{x}=(x_1, \dots, x_n)$. 
A variable of a monitor object is a shared variable if it is accessible by every 
thread that is accessing the monitor. The set of shared variables is denoted by 
$S$. The set of local variables, denoted by $L$, is 
accessible only by a thread calling a function in which the variables are declared. 

Predicates can be used to describe the properties of conditions. In our
approach,
every condition of {\em waituntil} statement is represented by a predicate. We say
a condition has been met if its representing predicate is true; otherwise, the
predicate is false. 
Furthermore, we assume that every predicate, $P = \vee_{i=1}^nc_i$, is in 
disjunctive normal form (DNF), where $c_i$ is defined as the conjunction of a 
set of atomic Boolean expressions. For example, a predicate $(x = 1) \wedge 
(y = 6) \vee (z \ne 8)$ is DNF, where $c_1 = (x = 1) \wedge (y = 6)$ and $c_2 = 
(z \ne 8)$. Note that, every Boolean formula can be converted into DNF using 
De Morgan's laws and distributive law. 

Predicates can be divided into two categories based on the type of their 
variables \cite{bh05}.
\begin{definition}[Shared and complex predicate]
    Consider a predicate $P(\vec{x}): X \rightarrow \mathbb{B}$. If $X 
    \subseteq S$, $P$ 
    is a shared predicate. Otherwise, it
    is a complex predicate. 
\end{definition}

The automatic-signal monitor has an efficient implementation \cite{kes77} by 
limiting the predicate of a {\em waituntil} to a shared predicate; however, 
we do not limit the predicate of a {\em waituntil} statement to a shared
predicate. The reason is that this limitation will lead {\em AutoSynch} to be less
attractive and practical since conditions including local variables cannot be 
represented in {\em AutoSynch}.

Evaluating a complex predicate in all threads is unattainable 
because the accessibility of the local variables in the predicate is limited 
to the thread declaring them. To evaluate a complex predicate in all 
threads, we treat local variables as constant values at runtime and define 
globalization as follows. 
\begin{definition}[Globalization]
    Given a complex predicate $P(\vec{x}, \vec{a}): X \times A \rightarrow 
    \mathbb{B}$, where $X \subseteq S$ and $A \subseteq L$. The globalization 
    of $P$ at runtime {\textbf t} is the new shared predicate
    \[
    G_t(\vec{x}) = P(\vec{x}, \vec{a_t}),
    \]
    where $\vec{a_t}$ is the values of $\vec{a}$ at runtime {\textbf t}. 
\end{definition}

The globalization can be applied to any complex predicate; a shared 
predicate can be derived from the globalization. For example, in
Fig.~\ref{fig:bb_exp}, the consumer $C$ wants to take $48$ items at some
instant of time. Applying the globalization to the complex predicate $(count \ge num)$
in line 19, we derive the shared predicate $(count \ge 48)$.

The following proposition shows that the complex predicate evaluation of
{\em waituntil} statement in all threads can be achieved through the globalization. 
\begin{proposition} \label{pro:glob}
    Consider a complex predicate $P(\vec{x}, \vec{a})$ in a \textbf
    {waituntil} 
    statement. $P(\vec{x}, \vec{a})$ and its globalization 
    $P(\vec{x}, \vec{a_t})$ are semantically equivalent during the \textbf{waituntil} 
    period, where $t$ is the time instant immediately before invoking the 
    \textbf{waituntil} statement.  
\end{proposition}
\begin{proof}
    Only the thread invoking the {\em waituntil} statement can access the
    local variables of the predicate; all other threads are unable to change
    the values of those local variables. Therefore, the value of $\vec{a}$
    cannot be changed 
    during the {\em waituntil} period. Since $\vec{a_t}$ is the value of $\vec{a}$
    immediately before invoking the {\em waituntil} statement, $P(\vec{x}, \vec{a})$
    and $P(\vec{x}, \vec{a_t})$ are semantic equivalent during the {\em waituntil}
    period. 
\end{proof}

Proposition \ref{pro:glob} enables the complex predicate evaluation of
{\em waituntil} statement in all threads. 
Given a complex predicate in a {\em waituntil} statement, in the sequel we substitute
all the local variables with their values immediately before invoking the
statement. The predicate can now be evaluated in all other threads during the
{\em waituntil} period. 

\subsection{Relay invariance} \label{sec:relay}
As mentioned in Section ~\ref{sec:sigAll}, {\em signalAll} calls are sometimes unavoidable
in the explicit-signal mechanism. In {\em AutoSynch}, {\em signalAll} calls are 
avoided by providing the {\em relay invariance}. 

\begin{definition}[Active and inactive thread] 
    Consider a thread that tries to access a monitor. If it is not waiting in a
    {\em waituntil} statement or has been signaled, then it is an {\em active} 
    thread for the monitor. Otherwise, it is an {\em inactive} thread. 
\end{definition}

\begin{definition}[Relay invariance]
    If there is a thread waiting for a predicate that is true, then there is at
    least one active thread; i.e., suppose $W_T$ is the set of waiting threads
    whose conditions have become true, $A_T$ is the set of active threads,
    then
    \[
        W_T \ne \phi \Rightarrow A_T \ne \phi 
    \]
    holds at all time. 
\end{definition}

{\em AutoSynch} uses the following mechanism for signaling.

{\em Relay signaling rule}: When a thread exits a monitor or goes into waiting 
    state, it checks whether there is some thread waiting on a condition that 
    has been true. If at least one such waiting thread exists, it signals that 
    thread.
    
\begin{proposition} \label{pro:relay_signal}
     The relay signaling rule guarantees relay invariance. 
\end{proposition}
\begin{proof}
    Suppose a thread $T$ is waiting on the predicate $P$ that is true. Since 
    $T$ is waiting 
    on $P$, $P$ must be false before $T$ went to waiting state. There must 
    exist another active thread $R$ after $T$ such that $R$ changed the state 
    of the monitor and made  $P$ true. According to the rule, $R$ must signal
    $T$ or another thread waiting for a condition that is true before leaving the
    monitor or going into waiting state. The thread signaled by $R$ then
    becomes active. Therefore, the relay invariance holds. 
 
 \end{proof}

The concept behind relay invariance is that, the privilege to enter the monitor
is transmitted from one thread to another thread whose condition has become true. 
For example, in Fig.~\ref{fig:bb_exp}, the consumer $C$ tries to take
$32$ items; however, only $24$ items are in the buffer at this moment. Then, $C$
waits for the predicate $P:  (count \ge 32)$ to be true. A producer, $D$, becomes 
active after 
$C$; $D$ puts $16$ items into the buffer and then leaves the monitor. Before 
leaving, $D$ finds that $P$ is true and then signals $C$; therefore, $C$ 
becomes active again and takes $32$ items of the buffer.
Proposition \ref{pro:relay_signal} shows that the relay
invariance holds in our automatic-signaling mechanism. 
Thus, {\em signalAll} calls are avoidable in {\em AutoSynch}. 
The problem is now reduced to finding a thread waiting for a condition that is 
true. 

\subsection{Predicate tag} \label{sec:tag}
In order to efficiently find an appropriate thread waiting for a predicate that is
true, we analyze every waiting condition and assign different tags 
to every predicate according to its semantics. These  tags
help us prune predicates that are not true by examining the state 
of the monitor. The idea behind the predicate tag is that, local variables cannot be
changed during the {\em waituntil} period; thus the values of local variables are
used as keys when we evaluate predicates. 
First, we define two types of 
predicates according to their semantics. 
\begin{definition}[Local and shared expression]
    Consider an expression $E(\vec{x}): X \rightarrow \mathbb{D}$, where
    $\mathbb{D}$ represents one of the primitive data types in Java. If $X
    \subseteq L$, then $E$ is a local expression. Otherwise, 
    if $X \subseteq S$, $E$ is a shared expression.  
\end{definition}
We use $SE$ to denote a shared expression, and $LE$ to denote a local 
expression.

\begin{definition}[Equivalence predicate]
    A predicate $P: (SE = LE)$ is an equivalence predicate.
\end{definition}
\begin{definition}[Threshold predicate]
   A predicate $P: (SE\ \boldsymbol{op}\ LE)$ is a threshold predicate, where 
   $\boldsymbol{op}
    \in \{<,\ \le,\ >,\ \ge\}$.
    \end{definition}
Note that, many predicates that are not equivalence or threshold predicates can 
be transformed into them. Consider the predicate $(x - a = y + b)$, where 
$x, y \in S$ and $a,b \in L$. This predicate is equivalent to 
$(x - y = a + b)$ which is an equivalence predicate. Thus, these two types of 
predicates can represent a wide range of conditions in synchronization problems. 

Given an Equivalence or a Threshold predicate, we can apply the {\em globalization} operation 
to derive a constant value on the right hand side of the predicate. 
%
In {\em AutoSynch}, there are three types of tags, $Equivalence$, $Threshold$, and 
$None$. Every $Equivalence$ or $Threshold$ tag represents an equivalence predicate 
or a threshold predicate, respectively. 
If the predicate is neither equivalence nor threshold, it acquires the $None$
tag. For example, consider the $Threshold$ predicate $x + b > 2y + a$ where $a$ 
and $b$ are local variables with values $11$ and $2$.
We first use the globalization to convert it to $(x - 2y > 9)$, which is
represented by the tag $(Threshold,\ x - 2y,\ 9,\ >)$. The formal definition of
tag is as follows. 
\begin{definition}
   A tag is a four-tuple $(M,\ expr,\ key,\ op)$, where  
   \begin{itemize}
      \item $M \in \{Equivalence,\ Threshold,\ None\}$;
      \item $expr$ is a shared expression if 
          $M \in \{Equivalence,\ Threshold\}$; otherwise, $expr= \perp$;
      \item $key$ is the value of a local expression after applying
          globalization if $M \in \{Equivalence,\ Threshold\}$; otherwise, 
          $key= \perp$;
      \item $op \in \{<,\ \le,\ >,\ \ge\}$ if $M = Threshold$; otherwise, 
         $op = \perp$.
   \end{itemize}
\end{definition}
We say that a tag is true (false) if the predicate representing the tag is 
true (false).
 
\subsubsection{Predicate tagging}
A tag is assigned to every conjunction. The tags of 
conjunctions of a predicate constitute the set of tags of the predicate. 
Tags are given to every predicate by the algorithm shown in
Fig.~\ref{fig:tagging}. 
When assigning a tag to a conjunction, the equivalence tag has the highest 
priority. The reason is that the set of values to make an equivalence predicate
true is smaller than the set of values to make a threshold predicate true. The equivalence predicate is true 
only when its shared expression equals a specific value. 
For example, consider an equivalence 
predicate $x = 8$ and a threshold predicate $x > 3$. The predicate $x = 8$ is true only when 
the value of $x$ is $8$, whereas $x > 3$ is true for a much larger set of values.
Therefore, the $Equivalence$ tags can help us prune
predicates that are false more efficiently than other kinds of
tags. 
If a conjunction does not 
include any equivalence predicate, then we check whether it 
includes any threshold predicate. If yes, then a $Threshold$ tag is assigned 
to the conjunction; otherwise, the conjunction has a $None$ tag. 

\begin{figure}[ht!]
    \begin{Verbatim}[fontsize=\footnotesize,gobble=8,frame=lines,
            framesep=3mm]
         tags = empty
         foreach conjunction c 
           if c contains an equivalence predicate se=le
             tag t = (Equivalence, se, globalization(le), null)
           else if c contains a threshold predicate se op le
             tag t = (Threshold, se, globalization(le), op)
           else 
             tag t = (None, null, null, null) 
           add t to tags 
         return tags
    \end{Verbatim}
  \caption{Predicate Tagging}
  \label{fig:tagging}
\end{figure}
Creating all tags for a conjunction is unnecessary. If a conjunction includes 
multiple equivalence predicates or threshold predicates, only one arbitrary 
$Equivalence$ tag or $Threshold$ tag is assigned to the conjunction. 
If there are a large number of tags, then the performance may decrease
because of the cost of maintaining tags. As a result, we assign only one tag to
every conjunction.
Assigning multiple tags to a 
conjunction cannot accelerate the searching process. For example, consider a 
conjunction $(x = 8) \wedge (y = 9)$. If only a tag 
$(Equivalence,\ x,\ 8,\ null)$
is assigned to the conjunction, we check the predicate when the tag is
true. Adding another tag $(Equivalence,\ y,\ 9,\ null)$ cannot accelerate the
searching process since we need to check both the tags. 
 
Note that multiple predicates with a shared conjunct may share a tag. For example,
the predicates $(x=5) \wedge (z \leq 4)$ and $(x=5) \wedge (y \geq 4)$ would have a shared equivalence tag
of $(x=5)$.

\subsubsection{Tag signaling}
Signaling mechanism is based on tags in {\em AutoSynch}. 
Since the equivalence tag is more efficient in pruning the search space than the threshold tag, the
predicates with equivalence are checked prior to the predicates with other 
tags. If no predicate that is true can be found after checking $Equivalence$ 
tags and $Threshold$ tags, our algorithm does the exhaustive search for the 
predicates with a $None$ tag. 

\paragraph{Equivalence tag signaling:}
Observes that, an equivalence predicate becomes true only when its shared 
expression equals the specific value of its local expression after applying
{\em globalization}. For distinct equivalence tags related to the same shared 
expression, at most one tag can be true at a time because the value of its
local expression is deterministic and unique at any time. By 
observing the value of its local expression, the appropriate tag can be 
identified. For example, suppose there are three $Equivalence$ tags for
predicates $x = 3$, $x = 6$, and $x = 8$. We examine $x$ and find that
its value is $8$. Then we know that only the third predicate $x = 8$ is true. Based on this 
observation, for each unique shared expression of an equivalence tag, we 
create a hash table, where the value of the local expression is used as the 
key. By using this 
hash table and evaluating the shared expression at runtime, we can find a
tag that is true in $O(1)$ time if there is any. Then we check the predicates 
with the tag. 


\paragraph{Threshold tag signaling:}
Consider the following example. Suppose there are two 
predicates, $x > 5$ and $x > 3$. We know that if $x > 3$ is false, then 
$x > 5$ cannot be true. Hence, we only need to check the predicate with the 
smallest local expression value for $>$ and $\ge$ operations. Furthermore, 
consider the predicates with the same local expression value but different 
operations, $x > 3$ and $x \ge 3$. The predicate $x > 3$ cannot be true when 
$x \ge 3$ is false; i.e., we only need to check the predicate $x \ge 3$. We use
a min-heap data structure for storing the threshold tags related to a same 
shared expression with $\boldsymbol{op} \in \{>, \ge\}$. If two predicates 
have the same local expression value but different operations, then the predicate 
with $\ge$ is considered to have a smaller value than the 
predicate with $>$ in the min-heap.
Similarly, the max-heap can be used for threshold tags with $\boldsymbol{op}
\in \{<, \le\}$. 



The signaling mechanism for $Threshold$ tag is shown if Fig.~\ref{fig:th_sig}. In
general, the tag in the root of a heap is checked. If the tag is false, all the
descendant nodes are also false. Otherwise, all predicates with the tag
need to be checked for finding a true predicate. To maintain the correctness, 
if no predicate is true, the tag is removed from the heap temporarily. Then the
tag in the position of the new heap root is checked again until a true predicate is found or a
false tag is found. Those tags removed temporarily are reinserted to the heap.
The reason to remove the tags is that the descendants of the tags may also be
true since the tags are true. So we also need to check the descendant tags. For
example, consider the predicates $P_1: (x \ge 5) \wedge (y \ne 1)$ and 
$P_2: (x > 7)$. $P_1$ has the tag $Q_1: (Threshold,\ x,\ 5,\ \ge)$ and $P_2$
has the tag $Q_2: (Threshold,\ x,\ 7,\ >)$. $Q_1$ is the root and $Q_2$ is its
descendant. Suppose at some time instant $x=3$, then $Q_1$ is false; thus, there is no
need to check $Q_2$. Now, suppose $x = 9$ and $y = 1$, then $Q_1$ is true. We
check all predicates that have tag $Q_1$.  Since $P_1$ is false, no predicate
having tag $Q_1$ is true. Then $Q_1$ is removed form the heap temporarily. We find the
new root $Q_2$ is true and $P_2$ that has tag $Q_2$ is also true. We signal 
a thread waiting for $P_2$ and then add $Q_1$ back to the heap. 

\begin{figure}[ht!]
    \begin{Verbatim}[fontsize=\footnotesize,gobble=8,frame=lines,
            framesep=3mm]
         // peek(): retrieve but does not remvoe the root 
         // poll(): retrieve and remove the root 
         list backup = empty;
         tag t = heap.peek();
         while t is true
           foreach predicate p with t
             if p is true
               signal a thread waiting on p
               foreach b in backup 
                 heap.add(b)
               return 
           backup.insert(heap.poll())
           t = heap.peek()
         foreach b in backup 
           heap.add(b)
    \end{Verbatim}
  \caption{Threshold tag signaling}
  \label{fig:th_sig}
\end{figure}

Suppose there are $n$ $Threshold$ tags for a shared expression with different 
keys. Suppose that these tags are assigned to $m$ predicates. The time 
complexity for maintaining the heap is $O(n \log(n))$ 
However, the performance is generally much better because we only
need to check the predicates of the tags in the root position in the most
cases. The time complexity for finding the root is $O(1)$. In the worst case, 
we need to check all predicates; thus, the time complexity is $O(n \log(n) + m)$. 
However, this situation is rare. Furthermore, this algorithm is optimized for
evaluating threshold predicates by sacrificing performance in tag management.

\section{{\em AutoSynch} implementation} \label{sec:imp}
The {\em AutoSynch} implementation involves two parts, the preprocessor and the Java
library of condition manager. The preprocessor,  built using JavaCC
\cite{kod04}, translates  {\em AutoSynch} code to Java code. Our
signal-mechanism is implemented in 
the condition manager library that creates condition variables, and maintains the association 
between predicates and condition variables. Furthermore, predicate tags are 
also maintained by the condition managers. It is the responsibility of the
condition manager to decide which thread should be signaled. 

\subsection{Preprocessor}
The {\em AutoSynch} class provides both mutual exclusion and conditional
synchronization. To maintain these two properties, our preprocessor adds some 
additional variables for any {\em AutoSynch} class. Fig.~\ref{fig:pre_cnst} 
summarizes the definitions of additional variables in the
constructor of an {\em AutoSynch} class. The lock 
variable, $mutex$, is declared for mutual exclusion, which is acquired at the 
beginning of every member function and released before the return statement.
In addition, a condition manager, $condMgr$, is declared for synchronization. 
The details of the condition manager are discussed in the next
section.

\begin{figure}[ht!]
    \begin{Verbatim}[fontsize=\footnotesize,gobble=8,frame=lines,
            framesep=3mm]
        Lock mutex
        ConditionManager condMgr 
        foreach shared predicate P
          tags = AnalyzePredicate(toDNF(P))
          condMgr.registerSharedPredicate(P, tags)
        foreach shared expression E
          condMgr.registerSharedExpression(E)
    \end{Verbatim}
  \caption{The additional variables for an {\em AutoSynch} class}
  \label{fig:pre_cnst}
\end{figure}  

All predicates are transformed to DNF in the preprocessing process by De Morgan's laws and 
distributive law. Then we analyze predicates to derive their tags. The condition 
manager registers the predicates and shared expressions for predicate 
evaluation. 
The shared 
predicates and shared expressions are identified in the preprocessing stage and added 
in the constructor of the class as in Fig.~\ref{fig:pre_cnst}. We add shared
predicates and shared expressions (but not complex predicates) in the construct because their semantics is
static and never changes. A complex predicate is registered dynamically
because its globalization may change according to the value of its local variables at
runtime. In Java, the shared 
predicates and shared expressions are created as inner classes that can access
the shared variables appearing in them with
$isTrue()$ or $getValue()$ functions for the condition manager to evaluate.
The function $isTrue()$ returns the evaluation of the shared predicate and
the function $getValue()$ returns the value of the shared expression.

For every member function of an {\em AutoSynch} class, the {\em mutex.lock()} and
{\em mutex.unlock()} are inserted at the beginning of the function and immediately before
the {\em return} statement, respectively.
%


In the {\em waituntil} statement, the predicate is checked initially. If 
it is true, then the thread can continue. Otherwise, the type of
predicate is checked. If the predicate is complex, then we apply globalization 
to it for deriving a new shared predicate. Then we query the condition manager 
to determine
whether the derived predicate has been added earlier. If not, we add the 
predicate with its tags to the condition manager. Then, the corresponding 
condition variable can be obtained by calling {\em getCondition()} function of the
condition manager. The {\em relaySignal()} function maintains relay invariance 
and signals an appropriate thread. Then, the thread goes into the waiting state 
until the predicate becomes true. After exiting the waiting state, if the 
predicate is complex and the corresponding condition has no other 
waiting thread, and then it is deactivated by the condition manager.

\begin{figure}[ht!]
    \begin{Verbatim}[fontsize=\footnotesize,gobble=8,frame=lines,
            framesep=3mm]
        if P is false 
          if P is a complex predicate 
            P := Globalization(toDNF(P))
            if P is not in condMgr
              tags = AnalyzePredicate(P)
              condMgr.registerComplexPredicate(P, tags)
          C = condMgr.getCondition(P)
          do 
            condMgr.relaySignal()
            wait C
          while P is false
          if P is complex predicate and C has no waiting thread
            condMgr.inactive(P) 
    \end{Verbatim}
    \caption{Preprocessing for a {\em waituntil(P)} statement}
  \label{fig:prep}
\end{figure}

%
%

\subsection{{\em AutoSynch} Java library: condition manager}
The condition manager maintains the predicates and condition variables, and
provides the signaling mechanism. To avoid creating redundant predicates and 
condition variables, predicates that have the same meaning should be mapped to 
the same condition variable. Two predicates are syntax equivalent if they 
are identical after applying globalization. A predicate table, which is implemented by a
hash table, records predicates and their associated condition variables. 

\begin{figure*}[ht!]
  \centering
  \includegraphics[width=180mm]{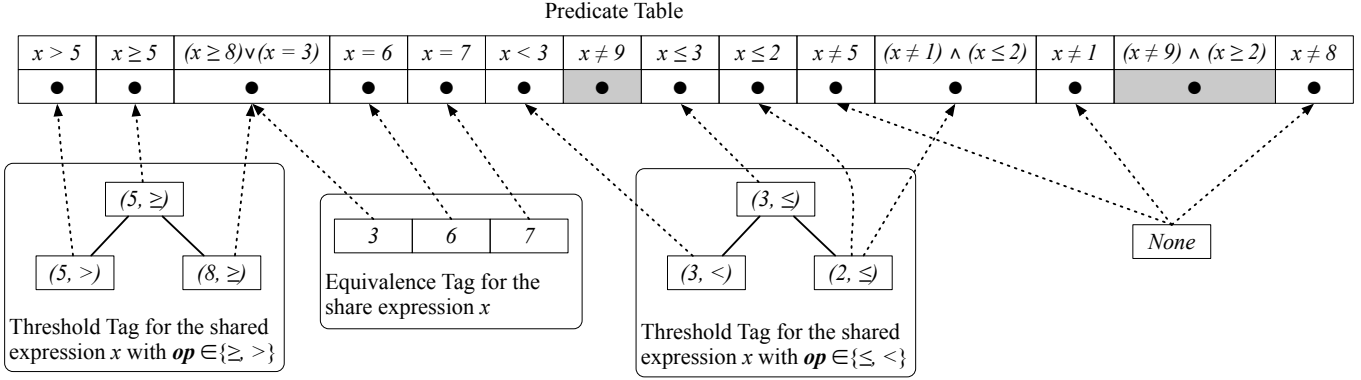}
  \caption{A example of the condition manager in {\em AutoSynch}}
  \label{fig:mgr}
\end{figure*}

When a predicate is added to the condition manager, 
its tags are stored in an appropriate data structure depending upon the type of its tag.
Fig.~\ref{fig:mgr} shows an example. The symbol $\bullet$ indicates 
a condition variable. The gray blank indicates that the predicate is inactive,
that is, no thread waits on it. A hash table is used for
storing equivalence tags with the shared expression $x$. In addition, 
a min-heap and a max-heap are used for storing threshold tags. 


For finding a predicate that is true in Fig.~\ref{fig:mgr}, the value of the
shared expression $x$ is evaluated. We first check the hash table 
(with $O(1)$ time complexity) using the value of the shared expression as the
key. If we find a tag in the hash table, then we evaluate predicates that have
the tag. If there exists a predicate is true, then we signal its corresponding 
condition 
variable. Otherwise, we check the max-heap and the min-heap. If we find that 
both tags in the roots are false, we search for the predicates with the $None$ 
tag exhaustively. If one of these predicates is true
we signal the corresponding condition variable. As can
be expected, the equivalence and threshold tags are helpful for searching 
predicates that are true.   

A predicate must be removed from the tag once no thread waits on 
it to avoid unnecessary predicate evaluation. A threshold tag also needs to be
removed once it has no predicate. 

Predicates may be reused. Instead  of removing those predicate with no waiting
thread, 
we move those predicate to an inactive list. If they are used later, then we 
remove them from the inactive list. Otherwise, when the length of the inactive list exceeds some 
predefined threshold, we remove the oldest predicates from the list. Note that, the shared
predicates are never removed since they are static and are added only at the
constructor. 

\section{Evaluation} \label{sec:eval}
We discuss the experiments conducted for evaluating the performance of {\em
AutoSynch} in this section. We compare the performances of different signaling
mechanisms in three sets of classical conditional synchronization problems. 
The first set of problems relies on only shared predicates for synchronization. 
Next, we explore the performance for problems using complex predicates. 
Finally, we evaluate the problems on which {\em signalAll} calls are required 
in the explicit-signal mechanism. 

\subsection{Experimental environment}
All of the experiments were conducted on a machine with 16 Intel(R) Xeon(R) 
X5560 Quad Core CPUs (2.80 GHz) and 64 GBs memory running Linux 2.6.18. 

Our experiments are {\em saturation} tests \cite{bh05}, in which only
monitor accessing function is performed. That is, no extra work is in the
monitor or out of the monitor. For every experimental setting, we  
perform 25 times, and remove the best and the worst results. Then we compare 
the average runtime for different signaling mechanisms.

\subsection{Signaling mechanisms}
Four implementations using different signaling mechanisms have been 
compared. 
\begin{description}
    \item[Explicit-signal] Using the original Java explicit-signal mechanism. 
    \item[Baseline] Using the automatic-signal mechanism relying on only
        one condition variable. It calls {\em signalAll} to wake
        every waiting thread. Then each waken thread re-evaluates its own 
        predicate after re-acquiring the monitor.
    \item[{\em AutoSynch-T}] Using the approach described in this paper but excluding
        predicate tagging. 
    \item[{\em AutoSynch}] Using the approach described in this paper. 
\end{description}

\subsection{Test problems}
Seven conditional synchronization problems are implemented for evaluating our
approach. 

\subsubsection{Shared predicate synchronization problems}
\begin{description}
    \item[Bounded-buffer \cite{dijk65, dijk71}] This is the traditional 
        bounded-buffer problem. Every producer waits if the buffer is full,
        while every consumer waits if the buffer is empty. 
    \item [Sleeping barber \cite{dijk65, dijk71}] The problem is analogous 
        to a barbershop with one barber. A barber has number of waiting
        chairs. Every time when he finishes cut, he checks whether some
        customers are waiting. If there are, he cuts hair for one customer. If 
        no customers waiting, the barber goes to sleep. Every customer arrives
        and checks what the barber is doing. If the barber is sleeping, then he 
        wakes the barber and has haircut. Otherwise, the customer checks
        whether there is any free waiting chair. If there is, the customer
        waits; otherwise, the customer leaves. 
    \item [$H_2O$ problem \cite{and00}] This is the simulation of water
        generation. Every $H$ atom waits if there is no $O$ atom or another $H$
        atom. Every $O$ atom waits if the number of $H$ atom is less than $2$. 
\end{description}
\subsubsection{Complex predicate synchronization problems}
\begin{description}
    \item[Round-Robin Access Pattern] Every test thread accesses the
        monitor in round-robin order. 
    \item[Readers/Writers  \cite{chp71}] 
    We use the approach given in \cite{bh05}, where a ticket is used
        to maintain the accessing order of readers and writers. Every reader
        and writer gets a ticket number indicating its arrival order. Readers
        and writers wait on the monitor for their turn. 
    \item [Dining philosophers \cite{dijk71}] A number of philosophers are 
        siting around at a table with a dish in front of them and a chopstick 
        in between each philosopher. A philosopher needs to pick two chopsticks at the
        same time for eating and he does not put down a chopstick until he finishes 
        eating. A philosopher that wants to eat must wait if one of his shared 
        chopsticks is hold by another philosopher.
\end{description}
\subsubsection{Synchronization problems required {\em signalAll} in explicit }
\begin{description}
    \item[Parameterized bounded-buffer   \cite{dijk65, dijk71}] The 
        parameterized bounded-buffer problem shown in Fig.~\ref{fig:bb_exp}. 
\end{description}

\subsection{Experimental results}

Fig.~\ref{fig:pc_eval} to \ref{fig:sb_eval} plot the results for the
bounded-buffer, the $H_2O$, and the sleeping barber problem. The y-axis shows
the runtime in seconds. The x-axis represents the number of simulating threads. 
Note that, in the $H_2O$ problem, only one thread simulating an $O$ atom. The
x-axis represents the number of thread simulating $H$ atoms. 
As expected, 
the baseline is much slower than other three signaling mechanisms, which have
similar performance in the bounded-buffer problem and the $H_2O$ problem. This 
phenomenon can be explained as follows. There is only a constant number of
shared predicates in {\em waituntil} statements for automatic-signal mechanisms.  
For example, in the bounded-buffer problem, there are two {\em waituntil} statements 
with global predicates, {\em count $>$ 0} (not empty condition) and {\em count
$<$ buff.length} (not full condition). Therefore, the 
complexity for signaling a thread in {\em AutoSynch} and {\em AutoSynch-T} is 
also constant. Hence, both {\em AutoSynch} and {\em AutoSynch-T} are as 
efficient as the explicit-signal mechanism. An interesting point is that the
performance of the baseline is as efficient as others in the sleeping barber
problem. The reason is that the {\em signalAll} calls of the baseline do not
increase the number of context switches. Whenever a signaled customer
re-acquires the monitor, he can have a haircut since the previous customer has
had haircut. These experiments illustrate that the automatic-signal mechanisms
are as efficient as the explicit-signal mechanisms for synchronization 
problems relying on only shared predicates. 

\begin{figure}[ht!]
  \centering
  \includegraphics[width=70mm]{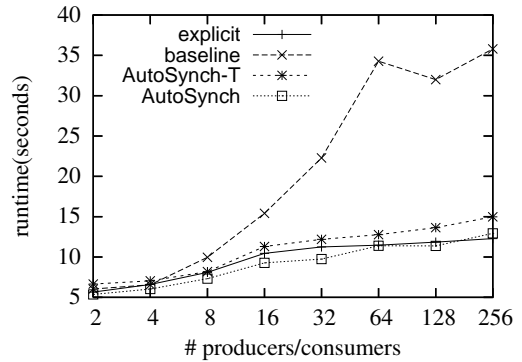}
  \caption{The results of bounded-buffer problem}
  \label{fig:pc_eval}
\end{figure}

\begin{figure}[ht!]
  \centering
  \includegraphics[width=70mm]{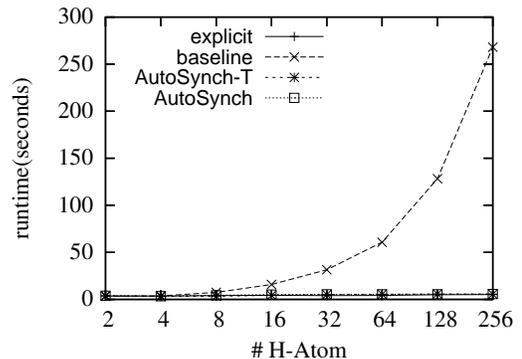}
  \caption{The results of $H_2O$ problem}
  \label{fig:h2o_eval}
\end{figure}

\begin{figure}[ht!]
  \centering
  \includegraphics[width=70mm]{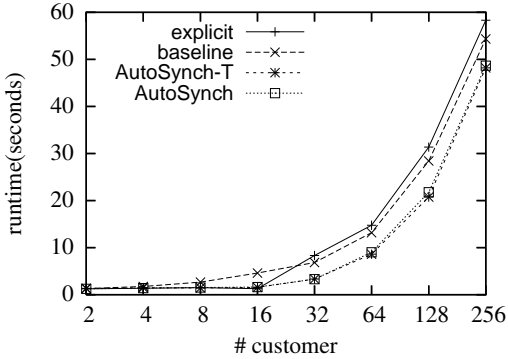}
  \caption{The results of sleeping barber problem}
  \label{fig:sb_eval}
\end{figure}

Fig.~\ref{fig:rr_eval} to \ref{fig:dp_eval} presents the experimental results for
the round-robin access pattern, the readers/writers problem, 
and the dining philosophers problem. The result of the
baseline is not plotted in these figures since its performance is extremely
inefficient in comparison to other mechanisms.  In this set of experiments, 
the explicit-signal mechanism has an advantage since it can explicitly signal 
the next thread to enter the monitor. For example, in the round-robin access
patter, an array of condition variables is used for associating the id of each
thread and its condition variable. Each thread waits on its condition variable
until its turn. When a thread leaves the monitor, it signals the condition 
variable of the next thread. As can be seen, the performance of 
explicit-signal mechanism is steady as the number of thread increases in
the round-robin access pattern and the reader/writers problem. In 
{\em AutoSynch-T}, its runtime increases significantly as the 
number of thread increase. For {\em AutoSynch}, the performance is slower than 
the explicit-signal mechanism between 1.2 to 2.6 times for the round-robin access
pattern. However, the performance of {\em AutoSynch} does not decrease as the 
number of threads increases. Note that, in the readers/writers problem, the
{\em AutoSynch-T} is more efficient than {\em AutoSynch} when the number of
threads is small. The reason is that {\em AutoSynch} sacrifices performance for
maintaining predicate tags. The benefit of predicate tagging increases as the
number of threads increases. Another interesting point is that the
performance of the explicit signal mechanism does not outperform implicit
signal mechanisms much in the dining philosophers problem. The reason is that
a philosopher only competes with two other philosophers sitting near him even
when the number of philosophers increases.

\begin{figure}[ht!]
  \centering
  \includegraphics[width=70mm]{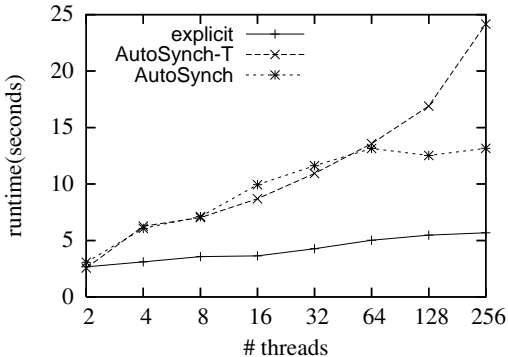}
  \caption{The results of round-robin access pattern}
  \label{fig:rr_eval}
\end{figure}

\begin{figure}[ht!]
  \centering
  \includegraphics[width=70mm]{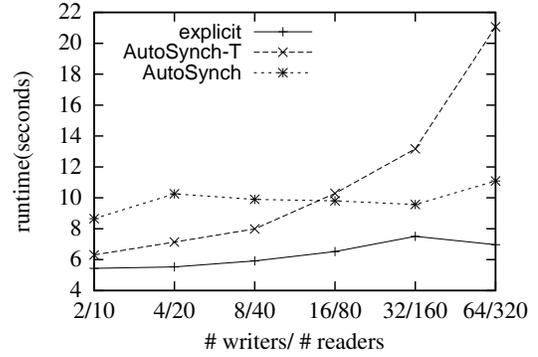}
  \caption{The results of readers/writers problem}
  \label{fig:rw_eval}
\end{figure}


\begin{figure}[ht!]
  \centering
  \includegraphics[width=70mm]{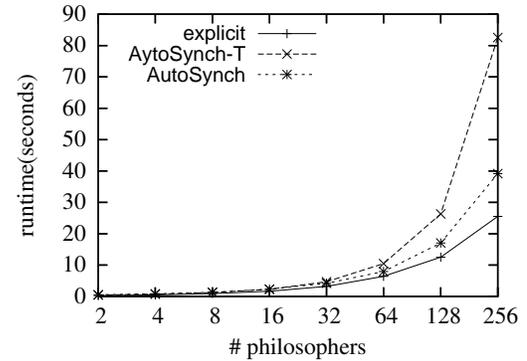}
  \caption{The results of dining philosophers problem}
  \label{fig:dp_eval}
\end{figure}

Table \ref{tab:cpu} presents the CPU usage (profiled by YourKit \cite{yourkit}) 
for the round-robin access pattern with 128 threads. The {\em relaySignal} is the
process of deciding which thread should be signaled in both {\em AutoSynch} and
{\em AutoSynch-T}. {\em Tag Mger} is the computation for
maintaining predicate tags in {\em AutoSynch}. 
As can be seen, the 
predicate tagging significantly improves the process for finding a predicate 
that is true. The CPU time of {\em relaySingal} process is reduced $95\%$ with
a slightly cost in tag management. 

\begin{table*}[ht!]
   \centering
   \begin{tabular}{|c||c|c||c|c||c|c||c|c|c|c|c|}
      \hline 
      & \multicolumn{2}{c||}{await} & \multicolumn{2}{c||}{lock} & 
        \multicolumn{2}{c||}{relaySignal} & \multicolumn{2}{c|}{Tag Mger} &
        \multicolumn{2}{c|}{others} & total \\
      \hline
         & T & \% & T & \% & T & \% & T & \% & T & \% & T \\
      \hline 
      \hline 
      explicit & $21365$ & $99.7\%$ & $28$ & $0.15\%$ & $NA$ & $NA$ & $NA$ &
      $NA$  & $28$ & $0.15\%$ & $21433$ \\
      \hline 
      {\em AutoSynch-T} & $410377$ & $98.5\%$ & $3140$ & $0.7\%$ & $2108$ & $0.5\%$
      & $NA$ & $NA$ & $1033$ & $0.2\%$ & $416658$\\
      \hline 
      {\em AutoSynch} & $96754$ & $98.8\%$ & $812$ & $0.8\%$ & $112$ & $0.1\%$ & 
      $124$ & $0.1\%$ & $148$ & $0.02\%$ & $97950$\\
      \hline 
   \end{tabular}
   \caption{The CPU usage for the round robin access pattern}
   \label{tab:cpu}
\end{table*}

In Fig.~\ref{fig:rpc_eval}, we compare the results of the parameterized 
bounded-buffer in which {\em signalAll} calls are required in the 
explicit-signal mechanism. In this experiment, there is one producer, which
randomly puts 1 to 128 items every time. The y-axis indicates the number of 
consumers. Every consumer randomly takes 1 to 128 items every time. As can be
seen, the performance of the explicit-signal mechanism decreases as the number
of consumers increases. {\em AutoSynch} outperforms the explicit-signal 
mechanism by 26.9 times when the number of thread is 256. This can be explained 
by Fig.~\ref{fig:csrpc_eval} that depicts the number of contexts switches. The
number of context switches increases in the explicit-signal
mechanism in which the number of context switches is around 2.7 millions when 
thread is 256. However, the numbers of context switches are stable in {\em
AutoSynch} even the number of threads increase. It has around 5440 context
switches when the number of thread is 256. This experiment demonstrates
that the number of context switches can be dramatically reduced and the
performance can be increased in {\em AutoSynch} for the problems required 
{\em signalAll} calls in the explicit-signal mechanism. 

\begin{figure}[ht!]
  \centering
  \includegraphics[width=70mm]{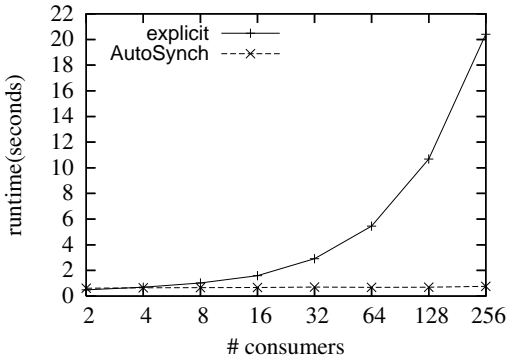}
  \caption{The results of the parameterized bounded-buffer problem}
  \label{fig:rpc_eval}
\end{figure}

\begin{figure}[ht!]
  \centering
  \includegraphics[width=70mm]{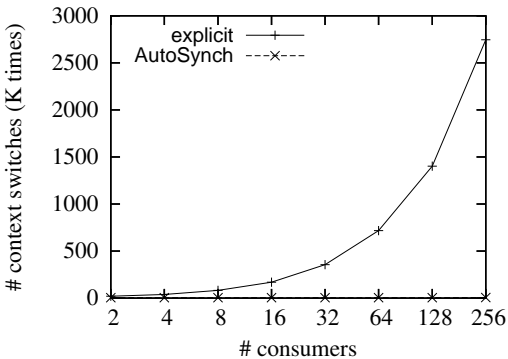}
  \caption{The number of context switches of the parametrized bounded-buffer 
  problem}
\label{fig:csrpc_eval}
\end{figure}

\section{Conclusions} \label{sec:conclu}
In this paper, we have proposed {\em AutoSynch} framework that supports 
automatic-signal mechanism with {\em AutoSynch} class and {\em waituntil} statement.
{\em AutoSynch} uses the {\em globalization} operation to enable the complex predicate 
evaluation in every thread. Next, it provides {\em relay invariance} that some
thread waiting for a condition has met is always signaled to avoid {\em signalAll}
calls. {\em AutoSynch} also uses predicate tag to accelerate the process in deciding
which thread should be signaled. 

To evaluate the effectiveness of {\em AutoSynch}, we built a prototype implementation
using JavaCC \cite{kod04}, Java Compiler Compiler,  and applied it to seven 
conditional synchronization problems. The experimental results indicate that 
{\em AutoSynch} implementations of these problems perform significantly better than
other automatic-signal monitors. Even though {\em AutoSynch} is around 2.6 times 
slower than the explicit in the worst case of our experiments, {\em AutoSynch} is
around 26.9 times faster than the explicit-signal in the parameterized 
bounded-buffer problem that relies on {\em signalAll} calls. 

In the future, we plan to optimize our framework through using the architecture
information. For example, we can get the number of cores of a machine, and
then limit the number of executing threads to avoid unnecessary contention. Our 
current implementation of {\em AutoSynch} is built upon constructs provided 
by Java. Thus, there is possibility of further performance improvement if the 
approach was to be implemented within the JVM. 

\section*{Acknowledgement}
We thank Himanshu Chauhan, Yen-Jung Chang, John Bridgman, and Craig Chase for 
the helpful comments to improve the paper.

\appendix

\end{document}